\documentclass{easychair}
\usepackage{qtree}
\usepackage[boxed,noend,linesnumberedhidden,figure]{algorithm2e}
\renewcommand{\k}{{\sf K}}

\newcommand*{\nat}{\ensuremath{\mathbb{N}}}

\usepackage{complexity}
\usepackage[boxed,noend]{algorithm2e}
\usepackage{mathpartir}

\renewcommand{\M}{\mathcal{M}}

\usepackage{amsmath}
\usepackage{amsthm}
\usepackage{amsfonts}
\usepackage{amssymb}

\newtheorem{theorem}{Theorem}
\newtheorem{lemma}{Lemma}

\title{Closing a Gap in the Complexity of Refinement Modal Logic}
\titlerunning{Closing a Gap in the Complexity of Refinement Modal Logic}
\author{Antonis Achilleos\inst{1} \and Michael Lampis\inst{2}\thanks{This research was partially supported by the Scientific Grant-in-Aid   from  Ministry of Education, Culture, Sports, Science and Technology of Japan.
}}

\institute{ Graduate Center, City University of New York\\ \email{aachilleos@gc.cuny.edu} 
	\and
	Research Institute for Mathematical Sciences (RIMS), Kyoto University\\ \email{mlampis@kurims.kyoto-u.ac.jp}
}
\authorrunning{Antonis Achilleos and Michael Lampis}

\begin{document}
\maketitle

\begin{abstract}

Refinement Modal Logic (RML), which was recently introduced by Bozzelli et al.,
is an extension of classical modal logic which allows one to reason about a
changing model. In this paper we study computational complexity questions
related to this logic, settling a number of open problems. Specifically, we 
study the complexity of satisfiability for the existential fragment of RML, a
problem posed by Bozzelli, van Ditmarsch and Pinchinat. Our main result is a
tight {\PSPACE} upper bound for this problem, which we achieve by introducing a
new tableau system. As a direct consequence, we obtain a tight characterization
of the complexity of RML satisfiability for any fixed number of alternations of
refinement quantifiers. Additionally, through a simple reduction we establish
that the model checking problem for RML is {\PSPACE}-complete, even for
formulas with a single quantifier.

\end{abstract}


\section{Introduction}

In this paper we deal with the computational complexity of the satisfiability
and model checking problems for the recently proposed Refinement Modal Logic
\cite{bozzelli2012}, an extension of the standard modal logic \k\  which also
adds ``refinement'' quantification. Specifically, we first settle the
complexity of satisfiability for the existential fragment of this logic,
resolving an open problem posed by Bozzelli, van Ditmarsch and Pinchinat
\cite{bozzelli2012b}. We then use this result to give tight characterizations
for the complexity of this problem for any fixed number of quantifier
alternations, closing the complexity gaps of \cite{bozzelli2012b}. In addition,
we use a simple argument to show that, unlike the standard modal logic \k,
model checking for RML is {\PSPACE}-hard even for formulas with a single
quantifier.

Modal logics with propositional quantifiers have been a subject of research
since Fine's seminal work \cite{fine1970}. The motivation behind the study of these
logics is their ability to model situations where new information is revealed
leading to a change in the model. In this context several bisimulation
quantified logics have been introduced recently
\cite{visser1996,hollenberg1998,pinchinat2007,french2006}. Roughly speaking,
they allow one to speak about models which are bisimilar to the current model,
except for a single (quantified-upon) propositional variable.  Thus, this
models the arrival of one kind of new information, namely new propositional
information.

The motivation behind Refinement Modal Logic on the other hand is to model the
arrival of a different kind of new information. We would like to have the
ability to speak about models which retain the propositional valuations of the
current model, but possibly differ from it in their structure. In other words,
we would like to be able to reason about new information which affects the
transitions, rather than the propositional valuations, of the current model.

A novel attempt to deal with this problem was offered by Balbiani et al.
\cite{balbiani2008knowable} who defined Arbitrary Public Announcement Logic
(APAL). Roughly speaking, their logic allows one to reason about models which
are restrictions of the current model. Unfortunately, this logic is undecidable
\cite{french2008}. A slightly more general (and arguably more natural) approach was then
undertaken by Bozzelli et al. \cite{van2009,van2010,bozzelli2012} with the introduction of
Refinement Modal Logic (RML). 

In RML one adds to the standard modal logic the quantifiers $\exists_r$ and
$\forall_r$. Informally, $\exists_r \phi$ means ``there exists a refinement of
the current model where $\phi$ holds''. This is more general than speaking
about restrictions, because a model $\M$ is a refinement of another model $\M'$
if and only if $\M$ is a restriction of \emph{a bisimilar copy of} $\M'$
\cite{bozzelli2012}. Another way to define refinement is as a partial
bisimulation: $\M$ is a bisimulation of $\M'$ if there exists a relation
between their states satisfying the \emph{atom}, \emph{forth} and \emph{back}
requirements. For a refinement, only the \emph{atom} and \emph{back} relations
need be satisfied. Informally, propositional valuations are retained and all
transitions of the refined model must map back to transitions of the original
model. As illustrated in \cite{bozzelli2012,bozzelli2012b} RML has already
found many applications, for example, in describing game strategies, program
logic or spatial reasoning.

The main purpose of this paper is to investigate the complexity of the
satisfiability problem for (single-agent) RML, continuing the work of Bozzelli
et al. \cite{bozzelli2012b}. An upper bound of 2\EXP\ was claimed for this
problem in \cite{van2010}, but the tableau-based procedure given there was
found to be sound but not complete \cite{bozzelli2012}. For the general case of
RML the problem was later settled in \cite{bozzelli2012b}, which showed that
the satisfiability problem is complete for the class A\EXP$_{poly}$, which is
the class of problems decidable by alternating exponential-time Turing machines
with a polynomial number of alternations. For any fixed quantifier-alternation
depth $k$, the same paper established that the problem is hard for the $k$-th
level of the exponential-time hierarchy and belongs in the $(k+1)$-th level.

Thus, this left a gap in the complexity of the problem for a fixed number of
alternations. This gap is mostly of theoretical interest for $k>0$, since (as
far as is currently known) there is no immediate practical difference in the
complexity of problems in the $k$-th and $(k+1)$-th level of the
exponential-time hierarchy: they all likely require doubly-exponential time.
This is not true though for the case of RML$^{\exists_r}$, which is the
fragment of RML using only existential quantifiers. The algorithm of
\cite{bozzelli2012} establishes membership of this problem in \NEXP\ (that is,
the first level of the exponential-time hierarchy), but the only hardness
result known is that following from the \PSPACE-hardness of ordinary modal
logic.

In this paper we settle this problem by giving a \PSPACE\ algorithm for
RML$^{\exists_r}$ satisfiability. In practical terms, this means that
satisfiability for this fragment is exponentially easier than for fragments
allowing universal quantification (unless \PSPACE=\NEXP). Somewhat
surprisingly, it also means that RML$^{\exists_r}$ satisfiability is no harder
than satisfiability for the standard modal logic \k, despite the large
expressive power given by adding existential refinement quantifiers. We also
show how our technique can be extended to close the complexity gaps of
\cite{bozzelli2012} for any fixed number of quantifier alternations, thus
completely settling the complexity of this problem. Our main ingredient is to
introduce a new tableau system for RML$^{\exists_r}$. We believe this
construction, which is arguably simpler than the constraint systems used in
\cite{bozzelli2012b}, may be of independent interest.

Our results above seem to suggest that adding existential refinement
quantifiers to \k\ does not cost much in terms of computational complexity.
Unfortunately, we complement this view with a negative result by showing that
the model checking problem is \PSPACE-hard for RML even for formulas with a
single quantifier, resolving an open problem from \cite{bozzelli2012}. The
reduction we propose is surprisingly simple but it establishes that in this
case the power afforded by a refinement quantifier causes an exponential
blow-up in the problem's complexity. 

The rest of this paper is structured as follows: in section \ref{sec:defs} we
give some preliminary definitions; in section \ref{sec:tableaux} we present the
tableau system for RML$^{\exists_r}$; in section \ref{sec:alg} we give a
\PSPACE\ algorithm for RML$^{\exists_r}$ using the tableau method and describe
how this idea can close the complexity gap for any fixed number of quantifier
alternations; finally in section \ref{sec:model} we present the complexity
results for the model checking problem.

\section{Preliminaries} \label{sec:defs}

We use standard terminology for modal logic and the definitions of RML given in
\cite{bozzelli2012,bozzelli2012b}. We assume an infinite supply of atomic
propositional variables. A Kripke structure is a triple $\M = (W,R,V)$, where
$W$ is the set of states (or worlds), $R\subseteq W\times W$ is a transition
relation and $V$ is a function which assigns to each state in $W$ a set of
propositional variables. To ease notation, when $(s,t)\in R$ we sometimes write
$sRt$. For two Kripke models $\M = (W,R,V)$ and $\M' = (W',R',V')$ we say that
$\M'$ is a refinement of $\M$ if there exists a non-empty relation $\mathcal{R}
\subseteq W\times W'$ such that:

\begin{itemize} 
	
	\item For all $(s,s')\in \mathcal{R}$ we have $V(s)=V'(s')$.
	
	\item For all $s\in W$, $s',t'\in W'$ such that $(s,s')\in\mathcal{R}$ and
	$s'R't'$ there exists $t\in S$ such that $(t,t')\in \mathcal{R}$ and $sRt$.
	
\end{itemize}

We call $\mathcal{R}$ a \emph{refinement} mapping from $\M$ to $\M'$. If the
inverse of $\mathcal{R}$ is also a refinement mapping from $\M'$ to $\M$ then
$\M$ and $\M'$ are bisimilar and $\mathcal{R}$ is a bisimulation. If
$\mathcal{R}$ is a refinement mapping from $\M$ to $\M'$ and for $s\in W$,
$s'\in W'$ we have $(s,s')\in\mathcal{R}$ we say that $(\M',s')$ is a
refinement of $(\M,s)$. 

The syntax of RML formulas (defined in \cite{bozzelli2012}) is as follows:

$$ \phi ::= p\ |\ \neg p\ |\ \phi \land \phi\ |\ \phi \lor \phi\ |\ \Diamond
\phi\ |\ \Box \phi\ |\ \exists_r \phi\ |\ \forall_r \phi$$

For most of the paper, unless otherwise noted, we will concentrate on the
existential fragment of RML, denoted RML$^{\exists_r}$, in which the
$\forall_r$ quantifier is not allowed. Observe that negations are only allowed
to appear immediately before propositional variables, making the distinction
between the existential and universal fragments meaningful. A formula that is
$p$ or $\neg p$ for some propositional variable $p$ is called a \emph{literal};
we usually denote literals by $l$. We use $d_\Diamond$ to denote the nesting
depth of modalitites $\Diamond,\Box$ and $d_\exists$ to denote the nesting
depth of $\exists_r$.

The semantics of the propositional and modal parts are as usual, and we write
$\M,s\models \phi$ if a formula $\phi$ holds in a state $s$ of a model $\M$.
For a formula $\phi=\exists_r \psi$ we have $\M,s\models \phi$ exactly if there
exists a model $\M'$ and a state $s'$ in that model such that $\M',s'\models
\psi$ and $(\M',s')$ is a refinement of $(\M,s)$. The semantics of $\forall_r$
are analogously defined; more details can be found in \cite{bozzelli2012}.

\section{Tableaux} \label{sec:tableaux}

In this section we introduce a tableau system that can be used to decide the
satisfiability of RML$^{\exists_r}$ formulas. In contrast to the constraint
systems introduced in \cite{bozzelli2012b}, we try to give a set of rules that
is as simple and as similar as possible to the standard tableau procedures
commonly used in ordinary modal logic \cite{d1999handbook}. 

Recall that, informally, the idea of a tableau is to maintain a set of prefix
and formula pairs. The prefixes, which are created by the tableau, represent
states of an attempted model of the initial formula. Usually, the prefixes are
strings of integers and the accessibility relation of the model we are building
can be read off the strings.

The main new idea in the tableau system we give here is that prefixes now
become pairs of strings. In the course of trying to satisfy the formula we are
implicitly constructing many models, because formulas of the form $\exists_r
\psi$ need to hold in a new model. At the same time, the new models we
construct must obey the refinement conditions. Therefore, to keep things tidy,
the tableau now produces formulas with two prefixes: the first prefix specifies
a model, with the requirement that the refinement relationship of models should
be readable from the prefixes, while the second, as before, encodes a state.

Let us now be a little more precise. We use elements of $\nat^*$ to encode
prefixes, which we denote by greek letters, such as $\mu,\nu,\lambda,\sigma$.
We use $\mu.\nu$ to denote concatenation. For such a string $\mu$, $|\mu|$ is
its length; elements of $\nat$ have length 1. The tableau contains elements of
the form $(\mu,\sigma)\ \psi$, where $\psi$ is a formula, $\mu$ is a prefix
that encodes the model and $\sigma$ a prefix that encodes the state. We
initialize the tableau with $(1,1)\ \phi$, where $\phi$ is the formula whose
satisfiability we want to decide. We now start applying some derivation rules,
which may need to create new states (in the case of $\Diamond$) or new models
(in the case of $\exists_r$). If $i\in\nat$ we use the convention that
$(\mu,\sigma.i)$ represents a state that is a successor of that represented by
$(\mu,\sigma)$. Similarly, we use the convention that $(\mu.i,\sigma)$
represents a state that is a refinement (in fact a restriction) of
$(\mu,\sigma)$.

Given this groundwork, the tableau rules follow quite naturally and are in fact
very similar to those for the standard modal logic K. Recall that the $\lor$
rule involves a non-deterministic choice and we will refer to a \emph{branch}
of the tableau as the set of elements produced after applying the rules to
$(1,1)\ \phi$, for some set of non-deterministic choices.

The full set of rules is given in Table \ref{tab:tableau}. Note that, perhaps the only somewhat
unusual part is the $\Box$ rule which can be applied to $(\mu,\sigma)\ \Box
\psi$ when $(\mu.\nu,\sigma.i)$ has appeared (as opposed to $(\mu,\sigma.i)$).
This is a clean way of maintaining the refinement relation: if
$(\mu.\nu,\sigma)$ represents a state that is a refinement of that represented
by $(\mu,\sigma)$ and $(\mu.\nu,\sigma.i)$ exists, then $(\mu,\sigma.i)$ must
also implicitly exist, even if it does not yet appear in the tableau.

\begin{table}[hb]

\begin{tabular*}{\textwidth}{@{\extracolsep{\fill} } ccc}
\begin{minipage}{0.20\textwidth}
\[
\inferrule*[right=$\wedge$]{(\mu,\sigma)\ \phi \wedge \psi}{(\mu,\sigma)\ \phi \\\\ (\mu,\sigma)\ \psi} 
\]
\end{minipage}&
\begin{minipage}{0.20\textwidth}
\[
\inferrule*[right=$\vee$]{(\mu,\sigma)\ \phi \vee \psi}{(\mu,\sigma)\ \phi\ \mid (\mu,\sigma)\ \psi} 
\]
\end{minipage}&
\begin{minipage}{0.20\textwidth}
\[
\inferrule*[right=$L$]{(\mu.\nu,\sigma)\ l}{(\mu,\sigma)\ l}\]
\end{minipage}\\
\begin{minipage}{0.20\textwidth}
\[
\inferrule*[right=$\Diamond$]{(\mu,\sigma)\ \Diamond \phi}{(\mu,\sigma.i)\ \phi} 
\]
\end{minipage}&
\begin{minipage}{0.20\textwidth}
\[
\inferrule*[right=$\exists_r$]{(\mu,\sigma)\ \exists_r \phi}{(\mu.m,\sigma)\ \phi} 
\]
\end{minipage}&
\begin{minipage}{0.20\textwidth}
\[
\inferrule*[right=$\Box $]{(\mu,\sigma)\ \Box \phi}{(\mu,\sigma.i)\ \phi}
\]
\end{minipage}\\
\begin{minipage}{0.20\textwidth}
where $\sigma.i$ has not appeared 
\end{minipage}
&
\begin{minipage}{0.20\textwidth} 
where $\mu.m$ has not appeared 
\end{minipage}
&
\begin{minipage}{0.20\textwidth}
where $(\mu.\nu,\sigma.i)$ has already appeared
\end{minipage}
\end{tabular*}

\caption{Summary of tableau rules.} \label{tab:tableau}

\end{table}

\medskip

A branch is complete if all rules have been applied as much as possible. A
complete branch is rejecting if there are some $(\mu,\sigma)\ p$ and
$(\mu,\sigma)\ \neg p$ both in the branch; otherwise it is accepting.  The
correctness of the tableau system is summarized in the following lemmata.

\begin{lemma}\label{lem:complete}
	If $\phi$ is satisfiable then there exists a complete accepting tableau branch starting from $(1,1)\phi$.
\end{lemma}
\begin{proof}
	We will proceed by induction on the number of rules applied.  During the course
	of the proof we maintain a collection of models.  Specifically, suppose that
	after $k \in \nat$ applications of the rules we have a tableau branch $b$. For
	every $\mu$ such that $(\mu,\sigma)\ \psi \in b$ we keep a model $\M_\mu =
	(W_{\mu},R_\mu,V_\mu)$. For each $(\mu,\sigma)$ such that $(\mu.\nu,\sigma)\
	\psi \in b$ for some $\nu\in\nat^*$ we have a state $s_{\mu,\sigma} \in W_\mu$.
	The following conditions are also satisfied:
	
	\begin{enumerate}
		
		\item Whenever $(\mu,\sigma)\ \psi \in b$, $\M_\mu, s_{\mu,\sigma} \models
		\psi$;
		
		\item for all $(\mu.\nu,\sigma)\ \psi, (\mu.\nu',\sigma.i)\ \psi' \in b$ with
		$\nu,\nu'\in\nat^*$, we have $s_{\mu,\sigma} R_\mu s_{\mu,\sigma.i}$;
		
		\item for every $(\mu.i.\nu,\sigma)\ \psi \in b$ with $\nu\in\nat^*$,
		$(\M_{\mu.i},s_{\mu.i,\sigma})$ is a refinement of
		$(\M_{\mu},s_{\mu,\sigma})$. Let $\rho_{\mu,\mu.i}$ be the refinement relation from $\M_{\mu}$ to $\M_{\mu.i}$. 
		For all $\nu\in\nat^*$ we define $\rho_{\mu,\mu.\nu.i}=
		\rho_{\mu,\mu.\nu}\circ \rho_{\mu.\nu,\mu.\nu.i}$ 
		
	\end{enumerate}
	
	The tableau is initialized with $(1,1)\phi$, and suppose that $\phi$ is
	satisfiable, so there exists a model $\M$ and a state $s$ in that model such
	that $\M,s\models \phi$. We set $\M_1=\M$ and $s_{1,1}=s$ and all the
	conditions above are satisfied.
	
	Suppose that after $k$ applications we have a collection of models that satisfy
	the conditions given above. We will show that we can handle applying any of the
	tableau rules, while still maintaining the conditions. This establishes that
	after any number of applications the tableau will be clash-free, since by the
	first condition, any formula contained in the tableau must be true in the
	corresponding state. We now take cases, depending on the rule applied.
The two propositional rules are simple. The $L$ rule is also simple since it
can only be applied to literals, whose values are preserved by refinement.  The
$\Box$ rule also causes no problem: suppose we apply it to $(\mu,\sigma)\
\Box\psi$ to produce $(\mu,\sigma.i)\ \psi$. Then by the conditions we have
that $\M_\mu,s_{\mu,\sigma}\models \Box \psi$ and, since some
$(\mu.\nu,\sigma.i)\ \psi'$ already appears in the tableau, there is a state
$s_{\mu,\sigma.i}\in W_\mu$ and $s_{\mu,\sigma} R_\mu s_{\mu,\sigma.i}$.
Therefore, $\M_\mu,s_{\mu,\sigma.i}\models \psi$ and the remaining conditions
are not affected.
	
	For the $\Diamond$ rule, suppose we apply it to $(\mu,\sigma)\ \Diamond\psi$ to
	get $(\mu,\sigma.i)\ \psi$, such that $\sigma.i$ has not appeared before in the
	branch. We know that $\M_\mu,s_{\mu,\sigma}\models \Diamond \psi$, therefore
	there exists a state in $W_\mu$, call it $s_{\mu,\sigma.i}$, such that
	$\M_\mu,s_{\mu,\sigma.i}\models \psi$, so the first condition remains
	satisfied. By the definition of refinement, if $\mu = m_1.m_2\ldots m_l$, then there must be some state in $W_{\mu'}$, which we call $s_{\mu',\sigma.i}$, such that $s_{\mu',\sigma.i}\rho_{\mu',\mu} s_{\mu,\sigma.i}$.
	Thus, we can continue with this procedure until the second and third conditions are
	satisfied for all prefixes.
	
	For the $\exists_r$ rule, suppose we apply it to $(\mu,\sigma)\ \exists_r \psi$
	to get $(\mu.m,\sigma)\ \psi$, such that $\mu.m$ has not appeared before in the
	branch. We know that $\M_\mu,s_{\mu,\sigma}\models \exists_r \psi$, therefore
	there exists a model, call it $\M_{\mu.m}$ and a state $s_{\mu.m,\sigma}$ such
	that $(\M_{\mu.m,\sigma},s_{\mu.m,\sigma})$ is a refinement of
	$(\M_\mu,s_{\mu,\sigma})$ and $\M_{\mu.m},s_{\mu.m,\sigma}\models \psi$. Thus, again, the
	two conditions are satisfied. 
\end{proof}

\begin{lemma}\label{lem:sound}
	If there exists a complete accepting tableau branch containing $(1,1)\ \phi$ then $\phi$ is satisfiable.
\end{lemma}
\begin{proof} 
	
	Suppose we have a complete accepting tableau branch $b$. We define for each
	$\mu$ such that $(\mu,\sigma)\ \psi\in b$ a model $\M_\mu=(W_\mu,R_\mu,V_\mu)$.
	Specifically, we set $W_\mu = \{ s_{\mu,\sigma}\ |\ (\mu.\nu,\sigma) \psi \in
	b\}$ and define $R_\mu$ such that for all $s_{\mu,\sigma},s_{\mu,\sigma.i}\in
	W_\mu$ we have $s_{\mu,\sigma} R_\mu s_{\mu,\sigma.i}$. We also set $p\in
	V_\mu(s_{\mu,\sigma})$ if and only if $(1,\sigma)\ p \in b$.
	
	Observe that for the models we have defined, $\M_{\mu.m}$ is a 
	refinement
	of
	$\M_\mu$, because every transition of the former model can also be found in the
	latter and also, by definition, corresponding states agree on their
	propositional variables.
	
	We will now prove by structural induction on the formula that whenever
	$(\mu,\sigma)\ \psi \in b$ we have $\M_\mu,s_{\mu,\sigma}\models \psi$. For
	propositional formulas and since we apply rule $L$ exhaustively without
	clashes, the assertion is straightforward. If $(\mu,\sigma)\ \Box \psi \in b$
	and the $\Box$ rule has been exhaustively applied, we know that for all
	$(\mu,\sigma.i)$ such that $(\mu.\nu,\sigma.i)$ appears in the tableau
	$(\mu,\sigma.i)\ \psi\in b$. By the inductive hypothesis this means that
	$s_{\mu,\sigma.i}\models \psi$. This covers all successors of $s_{\mu,\sigma}$
	in $\M_\mu$, therefore $\M_\mu,s_{\mu,\sigma}\models \Box \psi$.
	
	If $(\mu,\sigma)\ \Diamond \psi \in b$ and the $\Diamond$ rule has been applied
	then there exists a $(\mu,\sigma.i)\ \psi \in b$. By the inductive hypothesis,
	$\M_\mu,s_{\mu,\sigma.i}\models \psi$ therefore $\M_\mu,s_{\mu,\sigma}\models
	\Diamond \psi$.
	
	Finally, if $(\mu,\sigma)\ \exists_r \psi \in b$ and the $\exists_r$ rule has
	been applied, then there exists a model $\M_{\mu.m}$ and a state
	$s_{\mu.m,\sigma}$ such that, by inductive hypothesis
	$\M_{\mu.m},s_{\mu.m,\sigma}\models \psi$. Since, as observed $\M_{\mu.m}$ is a
	refinement
	of $\M_\mu$, we have
	$\M_{\mu},s_{\mu,\sigma}\models \exists_r \psi$.
\end{proof}

Let us also give a lemma that will prove useful in the analysis of the
algorithm of the following section. We will give an upper bound to the maximum
length of a prefix produced in any tableau branch.

\begin{lemma} \label{lem:musize}
	In any branch $b$ such that $(\mu,\sigma)\ \psi \in b$, we have $|\mu|\le d_\exists(\phi)$ and $|\sigma|\le d_\Diamond(\phi)$.
\end{lemma}
\begin{proof} 
	It is not hard to see that for all $(\mu.m, \sigma)\ \psi \in b$ there must
	exist $(\mu, \sigma)\ \exists_r \psi \in b$. Since $d_\exists(\exists_r
	\psi)>d_\exists(\psi)$, the first inequality follows. Similarly, for all $(\mu,
	\sigma.i)\ \psi \in b$ there must exist either a $(\mu, \sigma)\ \Diamond \psi
	\in b$ or $(\mu, \sigma)\ \Box \psi \in b$.
\end{proof}

\section{A Satisfiability Algorithm} \label{sec:alg}

In this section we describe an algorithm that decides satisfiability for the
existential fragment of RML. Our algorithm will make heavy use of the tableau
system defined in the previous section and will run in non-deterministic
polynomial space. Since it is known that \NPSPACE=\PSPACE
\cite{savitch1970,papadimitriou2003} and that the satisfiability problem is
\PSPACE-hard even for standard modal logic this algorithm gives a precise
complexity characterization for the satisfiability problem for
RML$^{\exists_r}$

The main idea of the algorithm is similar to the tableau procedure that decides
satisfiability for modal logic. We apply the tableau rules as defined in the
previous section and construct parts of a branch. Notice that we have bounded
the length of prefixes of the tableau formulas by the modal and existential
depths of the formula, so ideally we would want to perform DFS on these
prefixes. However, we need to be careful as there are two kinds of prefixes and
rules $L$ and $\Diamond$ when applied on a formula $(\mu,\sigma)\ \psi$ may
affect the prefixes of $\mu$.

The algorithm consists of the following: we keep a global counter variable (which we call $i$ and initialize to 1), to keep track of the different tableau prefixes we construct. We call the procedure RMLexSAT, which is described in Figure \ref{alg:rmlex}, with parameters $\{(1,1,\phi)\}$, $1$, $1$ and $\emptyset$. If the procedure returns, we declare the formula $\phi$ satisfiable. If at any point the procedure encounters a clash and rejects, the algorithm declares the formula unsatisfiable.

\begin{procedure}
	\KwIn{$P$, $\mu$, $\sigma$, $M\subseteq P$}
	\While{$P$ is not closed under the $\land,\lor$ and $L$ rules}{
	\ForEach{$(\nu,\sigma,\psi_1\wedge \psi_2)\in P \setminus M$}{
		$P:=P\cup \{(\nu,\sigma,\psi_1),(\nu,\sigma, \psi_2)\}$\;
		$M:=M\cup \{(\nu,\sigma,\psi_1\wedge \psi_2)\}$
		}
		\ForEach{$(\nu,\sigma,\psi_1\lor \psi_2)\in P\setminus M$}{
			Non-deterministically set $P:=P\cup \{(\nu,\sigma,\psi_1)\}$ or $P:=P\cup \{(\nu,\sigma,\psi_2)\}$\;
			$M:=M\cup \{(\nu,\sigma,\psi_1\lor \psi_2)\}$
		}
		\ForEach{$(\nu,\sigma,l)\in P\setminus M$, where  $l$ literal}{
			$P:=P\cup \{(\lambda,\sigma,l)\ |\ \lambda\sqsubseteq \nu\}$\;
			$M:=M\cup \{(\nu,\sigma,l)\}$
		}
	}
	\ForEach{$(\nu,\sigma,\Diamond \psi)\in P\setminus M$}{
		$P':=\{ (\lambda,\sigma.i,\chi)\ |\ (\lambda,\sigma,\Box \chi)\in P \textrm{ and } \lambda \sqsubseteq \nu   \} \cup \{(\nu,\sigma.i,\psi)\}$\;
		i:=i+1\;
		\nlset{A} \label{alg:firstcall}
		RMLexSAT($P'$,$\nu$,$\sigma.(i-1)$,$\emptyset$)\;
		$M:=M\cup\{(\nu,\sigma,\Diamond \psi)\}$
		}
	$N:=\{(\nu,\sigma,\exists_r \psi)\in P\setminus M\}$\;
	$M:=M\cup N$\;
	\ForEach{$(\nu,\sigma,\exists_r \psi)\in N$}{
		$P':=\{ (\lambda,\sigma,\chi)\in P\ |\ \lambda \sqsubseteq \nu   \} \cup \{(\nu.i,\sigma,\psi)\}$\;
		i:=i+1\;
		\nlset{B} \label{alg:seccall}
		$P'=$RMLexSAT($P'$,$\nu.(i-1)$,$\sigma$,$M$)\;
		$P := P \cup \{(\lambda,\sigma,l)\in P'| \lambda \sqsubseteq \nu, l $ literal$ \}$
	}
	\lIf{$P$ contains a clash}{Reject}			
		Return $P$
		\caption{RMLexSAT($P$, $\mu$, $\sigma$, $M$)} \label{alg:rmlex}
	\end{procedure}

\begin{theorem} \label{thm:alg}
	The satisfiability problem for RML$^{\exists_r}$ is in \PSPACE.
\end{theorem}

\begin{proof}
Let us first argue why the algorithm we described runs in non-deterministic polynomial space. 	First, we bound the recursion depth of the procedure by $d_\Diamond(\phi) \cdot d_\exists(\phi)$. 
 Assume RMLexSAT($P'$, $\mu'$, $\sigma'$, $M'$) was called inside the call of RMLexSAT($P$, $\mu$, $\sigma$, $M$). If that was the call in line \ref{alg:firstcall}, then notice that $|\sigma'|>|\sigma|$. Thus, in a recursion path, the number of recursive calls from line \ref{alg:firstcall} is bounded by $d_{\Diamond}(\phi)$(see lemma \ref{lem:musize}); it is enough to bound the depth of consecutive calls of the procedure from line \ref{alg:seccall} by $d_{\exists}(\phi)$. Consider such a subpath of consecutive recursive  calls from line \ref{alg:seccall}. Notice that as we go through this path only elements of $P\setminus M$ that are not of the form $(\nu,\tau,\Box \psi)$ contribute to the procedure. But after the first such call, say of RMLexSAT($P$, $\mu$, $\sigma$, $M$), any such element must be of the form $(\mu,\sigma,\psi)$. Thus, the next call along this this subpath, will further increase the length of $\mu$ and so on, so the subpath's depth is at most $|d_{\exists}(\phi)|$ (again, see lemma \ref{lem:musize}).
 
The elements in $P$ are at most $d_{\exists}(\phi) \cdot |\phi|$: the number of
subformulas of $\phi$ is at most $|\phi|$ and there are at most $|\mu|\le
d_\exists(\phi)$ prefixes of $\mu$. Thus, the branching of the recursion is at
most $|P|$ and the global variable $i$ is bounded by the total number of
recursive calls which in turn is at most $(d_\Diamond(\phi) \cdot
d_\exists(\phi))^{d_{\exists}(\phi) \cdot |\phi|}$. Therefore $i$ can be stored
using at most $d_{\exists}(\phi) \cdot |\phi| \cdot \log(d_\Diamond(\phi) \cdot
d_\exists(\phi))$ bits and this is enough to argue that our algorithm uses
space polynomial with respect to the size of its input.
 
 It is not hard to see that if there is an accepting tableau branch for $\phi$, then the algorithm can make appropriate non-deterministic choices to keep $P$ a set of triples corresponding to a subset of the branch and thus never run into a clash and accept. On the other hand, if there are appropriate non-deterministic choices for the algorithm to accept $\phi$, let $B$ be the set of all the triples ever produced during the run of the algorithm and let $b = \{(\mu,\sigma)\ \psi | (\mu,\sigma, \psi) \in B \}$. Notice that for every formula in $b$ and rule that can be applied on the formula, at some point the algorithm must have applied the rule on the formula, thus $b$ is a complete branch. It is also an accepting branch: if there is a clash, this means there are $(1,\sigma)\ p$ and $(1,\sigma)\ \neg p \in b$, thus at some point the algorithm produces $(1,\sigma, p)$ and at some other point $(1,\sigma, \neg p)$. Even if these points were during different recursive calls, notice that these triplets would transfer to the first recursive call with $\sigma$ as part of the input, which is a common ancestor in the recursion tree, due to the global variable $i$ which changes with every call.
\end{proof}

\paragraph{Complexity for fixed number of alternations.} Let us now sketch why
the above procedure can be used to close  the gaps in the complexity results  of \cite{bozzelli2012b}. Intuitively,
the idea is that since we can now solve satisfiability for RML$^{\exists_r}$ in
polynomial space (and hence deterministic exponential time) we save one
alternation for all the more complex cases.

The algorithm given in \cite{bozzelli2012b} is recursive. Given a formula
$\phi$, it (non-deterministically) guesses a model for $\phi$ of at most
exponential size. Then it checks whether the model actually satisfies $\phi$.
The algorithm proceeds to check if subformulas of $\phi$ hold in this model.
When $\forall_r \psi$ is encountered, the algorithm queries an oracle for the
truth of $\exists_r \neg \psi$ in the current model and if the answer is yes,
it rejects.  The resulting complexity comes from the fact that the alternation
depth of $\exists_r \neg \psi$ as defined in \cite{bozzelli2012b} is one less
than the one of $\forall_r \psi$.

The observation now is that we can save the last query access for $\exists_r
\psi$ if $\psi \in$ RML$^{\exists_r}$.  First notice that we can use our
algorithm to construct all complete accepting branches of the tableau in
exponential time.  We further argue that for any model for $\psi$, one of these
branches corresponds to a refinement of that model: notice that the proof of lemma \ref{lem:complete} works for any model for $\phi$ and the model constructed in the proof of lemma \ref{lem:sound} is a refinement of the original model for $\phi$ through the obvious refinement relation. The rest follows from the fact that refinement is a transitive relation.

\section{Model Checking} \label{sec:model}

In this section we deal with the complexity of the model checking problem for
RML, which was posed as an open question in \cite{bozzelli2012}. The main
(simple) observation we need is that, if we ignore propositional variables the
problem contains the satisfiability problem for \k. Since the satisfiability
problem for \k\ is \PSPACE-hard even for the fragment that contains no
variables \cite{ChagrovR02}, we get the same complexity for model checking for
RML and RML$^{\exists_r}$.

\begin{theorem}

Given an RML formula $\phi$, a model $\M$ and a state $s$ it is \PSPACE-hard to
decide if $\M,s \models \phi$ even if $\phi$ contains only a single refinement
quantifier. 

\end{theorem}

\begin{proof}

Suppose that we have a formula $\psi$ in the standard modal logic \k\ and that
$\psi$ contains no variables.  We will construct $\M,s$ and $\phi$ such that
$\M,s\models \phi$ if and only if $\psi$ is satisfiable.  Let $\M$ be the model
that consists of a single state $s$ and a transition from that state to itself.
It is not hard to see that any (variable-less) model is a refinement of $\M$ by
mapping all states to $s$. We let $\phi$ be $\exists_r \psi$. Now $\M,s \models
\phi$ if and only if there exists some model such that $\psi$ holds in one of
its states.
\end{proof}

For the sake of completeness, we will now aso demonstrate that model checking
for RML$^{\exists_r}$ is in {\PSPACE}. To do this
we use a variation of the tableau rules for
RML$^{\exists_r}$-satisfiability that we introduced above. Given a
model $\M$, state $u$ and formula $\phi$, we 
run the tableau procedure as was described in
section \ref{sec:tableaux} with the following restrictions: 
We say that $1$ is mapped to $u$ and 
whenever a prefix $(\mu,\nu.n)$ is introduced, $\nu$ is mapped
to $v$ and $\nu.n$ is new,
$\nu.n$ is mapped to some state accessible from $v$.
Whenever a prefixed formula $(\mu,\nu)\ l$ is introduced to
the branch, where $l$ a literal and $\nu$ is mapped to $v$, it must be
the case that $\M,v \models l$. If at some point it is impossible to
satisfy any of these restrictions, the branch is rejecting.
Finally, we have the additional rule:\\
	\begin{minipage}{0.20\textwidth}
		\[
		\inferrule*[right=$\Box_1$]{(1,\sigma)\ \Box \psi}{(1,\sigma.i_a)\ \psi}
		\]
	\end{minipage}
	\begin{minipage}{0.70\textwidth}
		where $\sigma.i_a$ is new and mapped to $a$, where $\sigma$ maps to $v$, $v R a$ and there is no $\sigma.j$  that maps to $a$.
	\end{minipage}
	
	The proof that the tableau procedure is correct and can be carried on using polynomial space is similar to the corresponding proofs in the previous sections. In particular, from an accepting branch it is not hard to see that we can construct a chain of refinements of $(\M,u)$ in a similar way as in lemma \ref{lem:complete} and using the extra conditions, such that $\M = \M_1$ and every $\sigma$ is mapped to $s_{1,\sigma}$ Then by induction on the formula we can show that $\M,u \models \phi$. On the other hand, if $\M,u \models \phi$, then it is straightforward to follow the proof of lemma \ref{lem:sound} to demonstrate that there is an accepting branch. Polynomial space follows from the fact that the number of states in the model is less than the size of the model, which is part of the input, thus it does not significantly affect the efficiency of the procedure. Therefore, we conclude:
	
	\begin{theorem}
		
		The model checking problem for RML$^{\exists_r}$ is \PSPACE-complete.
		
	\end{theorem}

\bibliographystyle{abbrv}
\bibliography{refinement}

\end{document}